\documentclass[pdflatex,sn-mathphys-num]{sn-jnl}

\usepackage{graphicx}
\usepackage{amsmath,amssymb,amsfonts}
\usepackage{amsthm}
\usepackage{mathtools}
\usepackage{mathrsfs}
\usepackage[title]{appendix}
\usepackage{placeins}
\usepackage{xurl}
\usepackage{booktabs}
\usepackage{tabularx}
\usepackage{xurl}
\usepackage{xcolor}
\usepackage{textcomp}
\usepackage{manyfoot}
\usepackage{etoolbox}
\usepackage{bm}
\usepackage{tikz}
\IfFileExists{glyphtounicode.tex}{\input{glyphtounicode}\pdfgentounicode=1}{}
\newcolumntype{Y}{>{\raggedright\arraybackslash}X}

\makeatletter
\patchcmd{\@maketitle}{Corresponding author(s). E-mail(s): }{E-mail: }{}{\ClassError{sn-jnl-patch}{e-mail line patch failed}{}}
\def\email#1{\global\advance\emailcnt by 1\relax%
\if@corauemail
   \g@addto@macro\corrauthemail{\setcounter{footnote}{0}\textcolor{blue}{#1}}%
\else
   \g@addto@macro\authemail{\setcounter{footnote}{0}\textcolor{blue}{#1}}%
\fi}
\makeatother

\newcommand{\vE}{\mathbf{E}}
\newcommand{\vD}{\mathbf{D}}
\newcommand{\vP}{\mathbf{P}}
\newcommand{\vB}{\mathbf{B}}
\newcommand{\vH}{\mathbf{H}}
\newcommand{\vM}{\mathbf{M}}
\newcommand{\vJ}{\mathbf{J}}
\newcommand{\vK}{\mathbf{K}}
\newcommand{\vp}{\mathbf{p}}
\providecommand{\vr}{}\renewcommand{\vr}{\mathbf{r}}
\newcommand{\vF}{\mathbf{F}}
\newcommand{\vq}{\mathbf{q}}
\newcommand{\nhat}{\hat{\mathbf{n}}}
\newcommand{\zhat}{\hat{\mathbf{z}}}
\newcommand{\rhat}{\hat{\mathbf{r}}}
\newcommand{\rhohat}{\hat{\boldsymbol{\rho}}}
\newcommand{\phihat}{\hat{\boldsymbol{\phi}}}
\newcommand{\thetahat}{\hat{\boldsymbol{\theta}}}
\newcommand{\qhat}{\hat{\mathbf{q}}}
\newcommand{\curl}{\nabla\times}
\newcommand{\dive}{\nabla\cdot}
\newcommand{\divs}{\nabla_{\!s}\cdot}
\newcommand{\vj}{\mathbf{j}}
\newcommand{\vk}{\boldsymbol{\kappa}}
\newcommand{\Jf}{\vJ_f}
\newcommand{\Jb}{\vJ_b}
\newcommand{\Jtot}{\vJ_{\mathrm{tot}}}
\newcommand{\JV}{\boldsymbol{\mathcal J}}
\newcommand{\rhotot}{\rho_{\mathrm{tot}}}
\newcommand{\Vm}{V_m}
\newcommand{\dd}{\,\mathrm{d}}

\makeatletter
\@ifundefined{th@thmstyleone}{%
\newtheoremstyle{thmstyleone}%
{18pt plus2pt minus1pt}
{18pt plus2pt minus1pt}
{\small\itshape}
{0pt}
{\small\bfseries}
{}
{.5em}
{\thmname{#1}\thmnumber{\@ifnotempty{#1}{ }\@upn{#2}}\thmnote{ {\the\thm@notefont(#3)}}}%
}{}
\makeatother
\theoremstyle{thmstyleone}
\newtheorem{proposition}{Proposition}

\allowdisplaybreaks
\AtBeginDocument{\ifdefined\bibsep\setlength{\bibsep}{0.7em}\fi}

\begin{document}

\title[Compensated Sources and Nonradiating Conditions]
{Compensated Free and Bound Sources and Nonradiating Conditions in Classical Electrodynamics}

\author*[1]{\fnm{Natan} \sur{Rentzber}}\email{nrentzbe@uccs.edu}

\affil*[1]{\orgdiv{Center for Magnetism and Magnetic Nanostructures},

\orgname{University of Colorado},
\orgaddress{\city{Colorado Springs}, \postcode{80918}, \country{USA}}}

\abstract{Classical macroscopic electrodynamics allows free and bound sources to cancel as distributions while the auxiliary equations retain nonzero free-source terms. A neutralized uniformly polarized sphere has $\vE=\mathbf 0$ and $\vD=\vP$, and a compensated uniformly magnetized sphere has $\vB=\mathbf 0$ and $\vH=-\vM$. These auxiliary fields are fixed by $\vP$ and $\vM$ rather than forming independent degrees of freedom. Time dependence separates charge cancellation from complete four-current cancellation. Canceling the charge alone leaves a divergence-free total current whose on-shell transverse transform controls radiation. For a neutralized polarized sphere of radius $R$ maintained by a tangential free-current sheet, the exterior field equals that of a point electric dipole whose effective moment is proportional to the spherical Bessel function $j_2(kR)$, where $k=\omega/c$. Every charge multipole and the ordinary magnetic dipole vanish, yet the source radiates at generic frequencies, and the complete exterior field vanishes at nonzero roots of $j_2$. Complete four-current cancellation removes all retarded source-generated fields. Among the compactly supported separable currents that maintain the same charge, it is the only choice that never radiates. In vacuum the auxiliary fields reduce to $\vD=\varepsilon_0\vE$ and $\vH=\vB/\mu_0$, so no radiation can be carried by $\vD$ and $\vH$ alone.}


\maketitle

\clearpage
\setcounter{tocdepth}{2}
\tableofcontents
\clearpage

\section{Introduction}
Classical macroscopic electrodynamics separates total charge and current densities into free and bound contributions. The bound sources are represented through the polarization $\vP$ and magnetization $\vM$, while the auxiliary fields are defined by~\cite{griffiths,jackson,hehl}
\begin{equation}
\vD=\varepsilon_0\vE+\vP,\qquad \vH=\frac{1}{\mu_0}\vB-\vM.
\label{eq:constit}
\end{equation}
The source equations may then be written as $\dive\vD=\rho_f$ and $\curl\vH=\Jf+\partial\vD/\partial t$. This leaves the fields $\vE$ and $\vB$ fixed by the total sources, while the values assigned to $\vD$ and $\vH$ depend on the material polarization, magnetization, and free-bound partition. The physical status of this source split and of the auxiliary fields has been discussed extensively~\cite{roche,gkm,hehl,mansuripur}.
\\

Exact compensation can be imposed between the free and bound sources. In the electric case, let
\\
\begin{equation}
\rho_f(\vr)+\rho_b(\vr)=0,
\label{eq:cancelrho}
\end{equation}
\\
with neither density identically zero. Gauss's law then gives $\dive\vE=0$, while the auxiliary equation retains $\dive\vD=\rho_f\not\equiv0$. In the magnetic case, take a magnetostatic bound current $\Jb=\curl\vM$ and impose
\begin{equation}
\Jf(\vr)+\curl\vM(\vr)=\mathbf 0,
\label{eq:cancelJ}
\end{equation}
again with neither term identically zero. The corresponding equations become $\curl\vB=\mathbf 0$ and $\curl\vH=\Jf\not\equiv\mathbf 0$. These relations raise two distinct issues. One concerns the fields that remain inside matter after exact source compensation. The other concerns whether a time-dependent compensated body can emit an auxiliary field into vacuum without the corresponding $\vE$ or $\vB$ field.
\\

The static and dynamic problems require different levels of cancellation. A local condition such as $\dive\vE=0$ or $\curl\vB=\mathbf 0$ does not determine the full field. Exact global cancellation must include every volume and surface contribution together with the boundary conditions. In the time-dependent case, $\rhotot=0$ implies $\dive\Jtot=0$ but does not imply $\Jtot=\mathbf 0$. A divergence-free total current may still have a nonzero on-shell transverse amplitude and radiate. Complete four-current cancellation sets both $\rhotot$ and $\Jtot$ to zero distributionally and removes the retarded source terms. Any field that reaches vacuum must satisfy $\vD=\varepsilon_0\vE$ and $\vH=\vB/\mu_0$, which rules out auxiliary-only radiation.
\\

The analysis develops exact compensated electric and magnetic bodies that realize these possibilities. A neutralized uniformly polarized sphere gives $\vE=\mathbf 0$ and $\vD=\vP$, while a compensated uniformly magnetized sphere gives $\vB=\mathbf 0$ and $\vH=-\vM$. Two time-dependent realizations of the polarized sphere separate charge cancellation from complete four-current cancellation. When the neutralizing charge is maintained by a tangential free-current sheet, the exact radiation-zone amplitude is proportional to $j_2(kR)$. The source radiates at generic frequencies although every charge multipole and its ordinary magnetic dipole vanish. At each nonzero root of $j_2(kR)$, the full outgoing exterior field vanishes by the exterior uniqueness theorem. The exterior field is in fact the exact point-dipole field of an effective moment proportional to $j_2(kR)$, and the completely canceled current is shown to be the unique maintaining current within the compactly supported separable class that never radiates. These results connect compensated macroscopic sources with the classical theory of radiating and nonradiating current distributions~\cite{ehrenfest,schott,goedecke,devaneywolf,gbur,abbottgriffiths}.
\\

A companion paper~\cite{twocurrents} utilizes the same example problem and solves the neutralized sphere as a boundary-value problem and obtains its interior fields, its energy balance, and its comparison with the bare polarized sphere. The analysis in this paper treats the auxiliary fields, the magnetized version, and the uniqueness question, and it reaches the same exterior coefficient by the on-shell transform.
\\

Section~\ref{sec:maxwell} fixes the source conventions and the distributional compensation rule. Sections~\ref{sec:questions} through~\ref{sec:radiation} establish the local, global, and radiative consequences. Section~\ref{sec:bodies} gives the explicit bodies and the exact $j_2(kR)$ result. The appendix derives the current transform used in the radiation calculation.

\section{Macroscopic Maxwell Equations and Auxiliary Fields}
\label{sec:maxwell}
SI units are used. On a material boundary, the unit normal $\nhat$ points from the material into the vacuum. The symbol $\Vm^c$ denotes the vacuum exterior of the material region $\Vm$, and $\divs$ denotes the surface divergence. The macroscopic Maxwell equations are
\begin{align}
\dive\vD&=\rho_f, & \dive\vB&=0,\\
\curl\vE&=-\frac{\partial\vB}{\partial t}, &
\curl\vH&=\Jf+\frac{\partial\vD}{\partial t},
\end{align}
together with Eq.~\eqref{eq:constit} and the bound-source definitions
\begin{equation}
\rho_b=-\dive\vP,\qquad \Jb=\frac{\partial\vP}{\partial t}+\curl\vM.
\label{eq:bound}
\end{equation}
The bound sources in Eq.~\eqref{eq:bound} have directly observable effects~\cite{herczynski}. The distinction between free and bound sources depends on the macroscopic material model and on the coarse-graining convention. Any reassignment that leaves $\rhotot$ and $\Jtot$ unchanged also leaves the resulting $\vE$ and $\vB$ unchanged.
\\

Equivalently, $\vB$ obeys $\curl\vB=\mu_0\Jtot+c^{-2}\partial\vE/\partial t$, where $\Jtot\equiv\Jf+\Jb$ and $c^2=1/(\mu_0\varepsilon_0)$. The equation for $\vH$ follows algebraically from this relation and Eq.~\eqref{eq:constit}. The terms $\curl\vM$ and $\partial\vP/\partial t$ have simply been moved into the definitions of $\vH$ and $\vD$. Likewise, $\dive\vD=\rho_f$ is Gauss's law $\dive\vE=\rhotot/\varepsilon_0$, with $\rhotot\equiv\rho_f+\rho_b$, after the bound charge $-\dive\vP$ is absorbed into $\vD$. Equation~\eqref{eq:constit} defines $\vD$ and $\vH$ once $\vP$ and $\vM$ are specified. It is not a constitutive relation. A constitutive law, such as $\vP=\varepsilon_0\chi_e\vE$ or $\vM=\chi_m\vH$ in a simple linear medium, specifies the material response. No such law is assumed here. In this macroscopic source description, the auxiliary equations do not add Lorentz-force fields beyond $\vE$ and $\vB$. The term auxiliary records that their values and source assignments depend on the chosen free and bound partition and on the material description. References~\cite{roche,gkm,hehl,mansuripur} discuss this structure in detail.
\\

All materials and interfaces are stationary in the working frame. Magnetic monopoles and dual currents are excluded, although ordinary surface currents and magnetization sheets are allowed. Volume fields are piecewise $C^1$ with the stated axis limits, and localized constructions decay at infinity. The infinite planar and cylindrical examples in Sec.~\ref{sec:analogs} are fixed by symmetry instead. Two standing assumptions specify the remaining field freedom. The static assumption excludes any externally imposed field, adopts the stated falloff or symmetry, and applies to every static result. The retarded prescription selects outgoing solutions with no incident field and no independent homogeneous field and applies to every dynamic result.

\FloatBarrier
\subsection{Distributional Sources and a General Compensation Rule}
The symbols $\Jf$ and $\Jb$ denote volume current densities, while $\rho_f$ and $\rho_b$ denote volume charge densities. Surface sheets are included distributionally. Let $\chi(\vr)$ be the indicator function of a stationary region $\Vm$, equal to $1$ inside and $0$ outside. Its gradient is $\nabla\chi=-\nhat\,\delta_S$, where $S\equiv\partial\Vm$. The surface delta function is defined by $\int f\,\delta_S\dd^3r=\oint_S f\dd S$. Fields $\vP$ and $\vM$ defined on the closure of $\Vm$ and extended by zero outside then carry the bound sources
\begin{align}
\rho_b&=-\dive(\vP\chi)=-(\dive\vP)\,\chi+(\vP\cdot\nhat)\,\delta_S,
\label{eq:rhobdist}\\
\Jb&=\frac{\partial(\vP\chi)}{\partial t}+\curl(\vM\chi)
=\Bigl(\frac{\partial\vP}{\partial t}+\curl\vM\Bigr)\chi+(\vM\times\nhat)\,\delta_S.
\label{eq:jbdist}
\end{align}
These expressions display the surface densities $\sigma_b=\vP\cdot\nhat$ and $\vK_b=\vM\times\nhat$ with their signs. Every later use of ``distributional'' has this meaning. Equations~\eqref{eq:rhobdist} and~\eqref{eq:jbdist} also give a general compensation rule. Choose the free sources
\begin{align}
\rho_f^{\mathrm{dist}}&=(\dive\vP)\,\chi-(\vP\cdot\nhat)\,\delta_S,\nonumber\\
\Jf^{\mathrm{dist}}&=\Bigl(-\frac{\partial\vP}{\partial t}-\curl\vM\Bigr)\chi-(\vM\times\nhat)\,\delta_S.
\label{eq:recipe}
\end{align}
The volume densities $\rho_f=\dive\vP$ and $\Jf=-\partial\vP/\partial t-\curl\vM$ inside $\Vm$, together with the sheets $\sigma_f=-\vP\cdot\nhat$ and $\vK_f=-\vM\times\nhat$, then cancel the bound four-current exactly. Therefore, $\rhotot=0$ and $\Jtot=\mathbf 0$ as distributions for any $\vP$ and $\vM$ in a fixed bounded region. The free four-current is conserved because the bound four-current satisfies
\begin{equation}
\frac{\partial\rho_b}{\partial t}+\dive\Jb
=-\dive\frac{\partial(\vP\chi)}{\partial t}
+\dive\Bigl[\frac{\partial(\vP\chi)}{\partial t}+\curl(\vM\chi)\Bigr]
=\dive\curl(\vM\chi)=0.
\end{equation}
The free distribution is the exact negative of the bound distribution and obeys the same continuity identity. Every construction below either follows Eq.~\eqref{eq:recipe} or departs from it in a stated way. In the volume constructions of Sec.~\ref{sec:volume}, $\vP$ and $\vM$ vanish on $\partial\Vm$, so no sheets appear. For the static uniform spheres, the bound charge or magnetization current lies entirely on the surface. The time-dependent electret also carries a volume polarization current. Its exact-cancellation version includes the matching free volume current. The slab and cylinder in Sec.~\ref{sec:analogs} are the infinite examples. The tangential-sheet electret in Sec.~\ref{sec:electret} cancels charge but deliberately leaves a nonzero total current.

\FloatBarrier
\subsection{Local and Global Cancellation}
Equations~\eqref{eq:cancelrho} and~\eqref{eq:cancelJ} concern volume densities inside $\Vm$. They imply $\dive\vE=0$ or $\curl\vB=\mathbf 0$ only in that region. A global statement about the primary field requires every volume and surface contribution to the total source to vanish on $\mathbb{R}^3$. An uncompensated surface charge $\sigma_f+\sigma_b\neq0$ or surface current $\vK_f+\vK_b\neq\mathbf 0$ can leave a field even when the volume densities cancel. The two sphere constructions in Sec.~\ref{sec:bodies} have exactly this surface-source structure. For a uniformly polarized or magnetized sphere in statics, the bound source lies on $r=R$, while $\rho_b=0$ and $\Jb=\mathbf 0$ in the interior. A matching free sheet neutralizes that surface source. These are boundary versions of the same apparent paradox, and the resolution must be applied distributionally across the surface.

\section{Local Compensated Source Relations}
The local consequences of Eqs.~\eqref{eq:cancelrho} and~\eqref{eq:cancelJ} follow directly from the macroscopic Maxwell equations.

\begin{proposition}[Electric]\label{prop:E}
Let $\Vm$ be a bounded material region in which Eq.~\eqref{eq:cancelrho} holds pointwise. The local relations inside $\Vm$ are
\\
\begin{equation}
\dive\vE=0,\qquad \dive\vD=\rho_f.
\end{equation}
\\
If $\rho_f\not\equiv0$ the two divergences are unequal.
\\
\end{proposition}
\begin{proposition}[Magnetic]\label{prop:B}
Let $\Vm$ be a bounded material region in magnetostatics. Assume that the bound current is $\Jb=\curl\vM$ and that Eq.~\eqref{eq:cancelJ} holds pointwise. Inside $\Vm$,
\\
\begin{equation}
\curl\vB=\mathbf 0,\qquad \curl\vH=\Jf.
\\
\end{equation}
If $\Jf\not\equiv\mathbf 0$ the two curls differ.
\end{proposition}

\noindent Both propositions follow by direct substitution. Equation~\eqref{eq:cancelrho} enters $\dive\vE=\rhotot/\varepsilon_0$ and $\dive\vD=\rho_f$. In magnetostatics, $\Jtot=\mathbf 0$ enters $\curl\vB=\mu_0\Jtot$, while $\curl\vH=\Jf$ remains unchanged.

\FloatBarrier
\subsection{Problem Statement and Main Results}\label{sec:questions}
The compensated relations lead to two separate problems. The first is whether $\vD$ or $\vH$ can remain nonzero inside matter when the corresponding $\vE$ or $\vB$ field vanishes, even though the auxiliary source equation retains a nonzero right-hand side. The second is whether a time-dependent compensated source can emit $\vD$ or $\vH$ into vacuum without the accompanying $\vE$ or $\vB$ field.
\\

Exact static solutions answer the first problem within matter. The remaining auxiliary field is $\vP$ or $-\vM$ and is fixed by the material source distribution. It is not an independent vacuum field. The second problem has a negative answer because the vacuum identities in Sec.~\ref{sec:conflation1} tie $\vD$ to $\vE$ and $\vH$ to $\vB$. The dynamic source problem then reduces to whether the total current vanishes or retains a radiative transverse component. The exterior field of the radiating realization equals a point-dipole field with the closed-form moment of Eq.~\eqref{eq:peff} at every frequency, and Proposition~\ref{prop:unique} shows the nonradiating maintainer to be unique within the compactly supported separable class.

\section{Vacuum and Material Regions}
\label{sec:conflation1}
Equation~\eqref{eq:constit} has different consequences in matter and vacuum. The relation $\vD=\varepsilon_0\vE$ holds wherever $\vP=\mathbf 0$, while $\vH=\vB/\mu_0$ holds wherever $\vM=\mathbf 0$. These conditions are independent. The difference between $\vD$ and $\varepsilon_0\vE$ is $\vP$, and the difference between $\vH$ and $\vB/\mu_0$ is $-\vM$. Both vacuum relations hold together where $\vP=\vM=\mathbf 0$, including the vacuum exterior.
\\

Suppose first that $\vE$ or $\vB$ vanishes inside $\Vm$. If $\vE=\mathbf 0$, then $\vD=\vP$ and $\dive\vD=\dive\vP=-\rho_b=\rho_f$, in agreement with Gauss's law for $\vD$. The displacement field remains nonzero because it equals $\vP$. In the magnetostatic setting of Proposition~\ref{prop:B}, with $\vP=\mathbf 0$, the condition $\vB=\mathbf 0$ gives $\vH=-\vM$ and $\curl\vH=\Jf=-\curl\vM$. The free volume or surface current and the boundary conditions determine this value. Thus an auxiliary field may remain nonzero where the corresponding $\vE$ or $\vB$ field vanishes, but it is fixed by $\vP$ or $-\vM$ and carries no independent vacuum radiation.
\\

In the vacuum exterior $\Vm^c$, the conditions $\vP=\vM=\mathbf 0$ give
\begin{equation}
\vD=\varepsilon_0\vE,\qquad \vH=\frac{1}{\mu_0}\vB \qquad (\vr\in\Vm^c).
\label{eq:vacuum}
\end{equation}
It follows that $\vD=\mathbf 0$ if and only if $\vE=\mathbf 0$, and $\vH=\mathbf 0$ if and only if $\vB=\mathbf 0$, in vacuum. Equation~\eqref{eq:vacuum} holds in any region where the model sets $\vP=\vM=\mathbf 0$. No electromagnetic solution in such a region can have $\vD\neq\mathbf 0$ with $\vE=\mathbf 0$, or $\vH\neq\mathbf 0$ with $\vB=\mathbf 0$. This result rules out an auxiliary-only vacuum field. The remaining issue is whether a charge-canceled body can radiate through a nonzero total current.

\section{Local Differential Constraints and Global Fields}
\label{sec:conflation2}
Neither $\dive\vE=0$ nor $\curl\vB=\mathbf 0$ determines an entire vector field. At a fixed time, the Helmholtz decomposition separates a localized field into longitudinal and transverse parts,
\\
\begin{equation}
\vF=\vF_L+\vF_T,\qquad \curl\vF_L=\mathbf 0,\quad \dive\vF_T=0.
\end{equation}
\\
This is an instantaneous kinematic decomposition. For smooth fields with sufficient decay, it is the standard Helmholtz theorem~\cite{jackson,kobe,arfken,herasheras}. The explicit reconstruction is
\\
\begin{equation}
\vF_L(\vr)=-\nabla\!\int\frac{\nabla'\!\cdot\vF(\vr')}{4\pi|\vr-\vr'|}\,\dd^{3}r',\qquad
\vF_T(\vr)=\curl\!\int\frac{\nabla'\!\times\vF(\vr')}{4\pi|\vr-\vr'|}\,\dd^{3}r',
\label{eq:helmholtz}
\end{equation}
\\
so the longitudinal part is built from the divergence alone and the transverse part from the curl alone. Surface sources can be included distributionally. The decomposition is used here only for compactly supported source currents and for the Coulomb component in Eq.~\eqref{eq:EL}. The outgoing $1/r$ field is obtained separately from the retarded solution in Sec.~\ref{sec:radiation}.
\\

\emph{Electric Case.} Define the instantaneous Coulomb component $\vE_L$ by Eq.~\eqref{eq:EL} and set $\vE_T\equiv\vE-\vE_L$. Since $\dive\vE_T=0$,
\begin{equation}
\dive\vE=\dive\vE_L=\frac{\rhotot}{\varepsilon_0},
\end{equation}
including any surface charge distributionally. The longitudinal component is the instantaneous Coulomb field of the total charge,
\begin{equation}
\vE_L(\vr,t)=-\nabla\,\frac{1}{4\pi\varepsilon_0}\int\frac{\rhotot(\vr',t)}{|\vr-\vr'|}\dd^3r',
\label{eq:EL}
\end{equation}
so $\rhotot=0$ on all of $\mathbb{R}^3$ gives $\vE_L=\mathbf 0$ immediately. No decomposition of the $1/r$ radiation field is needed. Gauss's law does not constrain the transverse component $\vE_T$, which is generally nonzero because Faraday's law gives $\curl\vE_T=-\partial\vB/\partial t$. A local divergence condition is weaker still. A field can be divergence-free in one subregion while containing a harmonic gradient produced by charges outside that region, since the longitudinal and transverse split is global and nonlocal. Global cancellation of free and bound charge therefore removes only the longitudinal Coulomb component of $\vE$. It does not remove the transverse component, which contains the radiative field. A uniform field and a vacuum plane wave are simple nonzero examples with zero divergence. They lie outside the localized-source class used in the global argument, but they directly disprove the local inference $\dive\vE=0\Rightarrow\vE=\mathbf 0$.
\newpage

\emph{Magnetic Case.} The magnetic field is solenoidal, $\dive\vB=0$. In a bounded region where $\curl\vB=\mathbf 0$, the field is locally both curl-free and divergence-free. In a simply connected subregion it can be written as $\vB=-\nabla\Psi_B$, where $\Psi_B$ is harmonic. This field need not vanish. Its value is fixed by boundary data and by currents outside $\Vm$. A zero local curl removes the local current source, not the field itself. External sources, surface currents on $\partial\Vm$, and boundary data can still determine a nonzero interior field. A divergence or curl condition alone cannot determine a three-component field. When $\rhotot=0$ globally, Gauss's law removes the Coulomb component of $\vE$ but leaves its transverse component and any locally harmonic contribution untouched. Likewise, $\curl\vB=\mathbf 0$ removes the local current source rather than $\vB$ itself. The compensated magnetized sphere in Sec.~\ref{sec:bodies} is a special case in which the external and surface contributions also vanish. Only the full global argument then gives $\vB=\mathbf 0$.

\section{Total Current and Radiation}
\label{sec:radiation}
The radiation problem is controlled by the time-dependent transverse part of the total current rather than by the separate appearance of free or bound source terms. Let the sources vary in time while the full Maxwell equations are retained. Both compensated cases can then be described through $\Jtot$.
\\

In the electric case, charge cancellation does not directly fix $\Jtot$. The continuity equation supplies the missing relation. The free and bound currents obey
\begin{equation}
\dive\Jf+\frac{\partial\rho_f}{\partial t}=0,\qquad
\dive\Jb+\frac{\partial\rho_b}{\partial t}=0.
\end{equation}
The second equation follows from $\rho_b=-\dive\vP$ and $\Jb=\partial\vP/\partial t+\curl\vM$, together with $\dive(\curl\vM)=0$. Adding the two continuity equations gives $\partial_t(\rho_f+\rho_b)+\dive\Jtot=0$. If the global distributional form of Eq.~\eqref{eq:cancelrho} holds at every instant, then
\begin{equation}
\dive\Jtot=0.
\label{eq:divJ}
\end{equation}
The condition $\rhotot=0$ will be called charge cancellation, while $\Jtot=\mathbf 0$ will be called total current cancellation. Their combination is complete four-current cancellation. In the magnetic case, take $\vP=\mathbf 0$, so $\Jb=\curl\vM$. Maintaining Eq.~\eqref{eq:cancelJ} then makes the total volume current vanish inside $\Vm$, and Eq.~\eqref{eq:divJ} follows. If $\vP\neq\mathbf 0$, the same result requires the stronger prescription $\Jf=-\partial\vP/\partial t-\curl\vM$, which is the current part of Eq.~\eqref{eq:recipe}. Distributional cancellation must also include every surface sheet. When it does, $\Jtot=\mathbf 0$ globally. Canceling only the volume currents can leave uncompensated surface currents that still source fields, as the magnetized sphere in Sec.~\ref{sec:bodies} shows. In the charge-cancellation case, Eq.~\eqref{eq:divJ} makes the longitudinal Helmholtz component of $\Jtot$, defined by Eq.~\eqref{eq:helmholtz}, divergence-free. That component is also curl-free by definition and has the same falloff, so it vanishes. Thus $\Jtot=\Jtot^{(T)}$ is purely transverse. Total current cancellation is stronger because it sets $\Jtot$ itself to zero.

Taking the curl of Faraday's law, using $\curl\vB=\mu_0\Jtot+c^{-2}\partial\vE/\partial t$, and substituting $\dive\vE=\rhotot/\varepsilon_0$ gives
\begin{equation}
\nabla^2\vE-\frac{1}{c^2}\frac{\partial^2\vE}{\partial t^2}
=\mu_0\frac{\partial\Jtot}{\partial t}+\frac{1}{\varepsilon_0}\nabla\rhotot,\qquad
\nabla^2\vB-\frac{1}{c^2}\frac{\partial^2\vB}{\partial t^2}=-\mu_0\,\curl\Jtot.
\label{eq:wave}
\end{equation}
Under charge cancellation, $\rhotot=0$ distributionally, so the charge term in the $\vE$ equation disappears. Equation~\eqref{eq:divJ} leaves only the transverse source $\mu_0\,\partial\Jtot^{(T)}/\partial t$. Under complete four-current cancellation, both source terms in Eq.~\eqref{eq:wave} vanish. The retarded prescription of Sec.~\ref{sec:maxwell} then gives $\vE=\vB=\mathbf 0$. The interior auxiliary fields $\vD=\vP$ and $\vH=-\vM$ may still remain. In Coulomb gauge, the charge density fixes the longitudinal electric field instantaneously, while the transverse current drives the vector potential and the dynamical transverse field~\cite{cohentannoudji}.
\\

The separate cancellation conditions in Eqs.~\eqref{eq:cancelrho} and~\eqref{eq:cancelJ} depend on the reference frame. A Lorentz boost mixes charge and current within each conserved four-current and also mixes $\vP$ with $\vM$, but it does not convert bound sources into free sources~\cite{hehl,ll,mansuripur}. Complete four-current cancellation, $J^\mu_{\mathrm{tot}}=(c\rhotot,\Jtot)=0$ distributionally, is frame independent. The vacuum identities in Eq.~\eqref{eq:vacuum} are also frame independent, so an auxiliary-only vacuum wave remains impossible in every inertial frame. A boost sets a material boundary in motion, and the explicit static solutions are therefore left in their working frame.
\\

A time-dependent transverse current can radiate, but it need not do so. An exact criterion follows from monochromatic sources in the time domain. Each construction below has the separable form $\vJ(\vr,t)=g(t)\,\vj(\vr)$, where one sinusoidal factor $g(t)$ multiplies a real distributional profile $\vj(\vr)=\vj_V(\vr)\,\chi_{r<R}+\vk(\vr)\,\delta_S$. The functions $\vj_V$ and $\vk$ are the volume and sheet profiles. The sheet profile carries one factor of length, as required for a surface density. Define the spatial cosine transform over all space by
\begin{equation}
\JV(\vq)=\int \vj(\vr)\,\cos(\vq\cdot\vr)\,\dd^{3}r=\int_V \vj_V(\vr)\,\cos(\vq\cdot\vr)\dd V+\oint_S\vk(\vr)\,\cos(\vq\cdot\vr)\dd S,
\label{eq:transform}
\end{equation}
where $\qhat=\vq/q$ for $q\neq0$. The corresponding sine transform vanishes because the profiles used here are even under $\vr\to-\vr$, as shown in Appendix~\ref{app:jq}. The potentials are taken in Lorenz gauge, where $\vE=-\nabla\Phi-\partial\mathbf A/\partial t$ and $\vB=\curl\mathbf A$ with $\dive\mathbf A+c^{-2}\,\partial\Phi/\partial t=0$. The macroscopic equations with total sources then decouple into the driven wave equations
\begin{equation}
\nabla^{2}\Phi-\frac{1}{c^{2}}\frac{\partial^{2}\Phi}{\partial t^{2}}=-\frac{\rhotot}{\varepsilon_0},\qquad
\nabla^{2}\mathbf A-\frac{1}{c^{2}}\frac{\partial^{2}\mathbf A}{\partial t^{2}}=-\mu_0\Jtot,
\label{eq:lorenz}
\end{equation}
and the retarded prescription selects the outgoing particular solution
\begin{equation}
\mathbf A(\vr,t)=\frac{\mu_0}{4\pi}\int\frac{\Jtot\bigl(\vr',\,t-|\vr-\vr'|/c\bigr)}{|\vr-\vr'|}\,\dd^{3}r',
\label{eq:retA}
\end{equation}
with the same form for $\Phi$ driven by $\rhotot/\varepsilon_0$. For a sinusoidal factor, $g(t_r+\nhat\cdot\vr'/c)=g(t_r)\cos(k\nhat\cdot\vr')+[g'(t_r)/\omega]\sin(k\nhat\cdot\vr')$, where $t_r\equiv t-r/c$, $g'\equiv\dd g/\dd t$, and $\nhat$ is the observation direction. The sine term integrates to zero against an even profile. Expanding $|\vr-\vr'|=r-\nhat\cdot\vr'+O(r'^{\,2}/r)$ in the radiation zone, as in the standard treatment~\cite{jackson}, then gives
\begin{equation}
\mathbf A(\vr,t)=\frac{\mu_0}{4\pi r}\,g(t-r/c)\,\JV(k\nhat)+O(r^{-2}),
\qquad k=\frac{\omega}{c}.
\end{equation}
For the compensated constructions below, $\rhotot=0$ distributionally at every time. The retarded scalar potential of Eq.~\eqref{eq:lorenz} therefore vanishes, and $\vE=-\partial\mathbf A/\partial t$ everywhere. Equation~\eqref{eq:divJ} also makes $\Jtot$ divergence-free. Its transform is transverse because $\vq\cdot\JV(\vq)=\int\vj\cdot\nabla\sin(\vq\cdot\vr)\,\dd^3r=-\int\sin(\vq\cdot\vr)\,\dive\vj\,\dd^3r=0$, with the boundary term vanishing for a localized source. The leading $1/r$ field is therefore already transverse. For a general conserved monochromatic source, the $1/r$ part of the retarded scalar potential cancels the longitudinal part of $-\partial\mathbf A/\partial t$ and leaves the same transverse projection~\cite{jackson}. At leading order in $1/r$, this gives $\vE_{\mathrm{rad}}=-\partial\mathbf A_\perp/\partial t$ and $\vB_{\mathrm{rad}}=\nhat\times\vE_{\mathrm{rad}}/c$,
\begin{equation}
\vE_{\mathrm{rad}}(\vr,t)=-\frac{\mu_0\,g'(t-r/c)}{4\pi r}
\Bigl[\JV(k\nhat)-\nhat\bigl(\nhat\cdot\JV(k\nhat)\bigr)\Bigr],
\qquad \vB_{\mathrm{rad}}=\frac{1}{c}\,\nhat\times\vE_{\mathrm{rad}}.
\label{eq:onshell}
\end{equation}
The remaining terms are $O(r^{-2})$. For a general monochromatic source, write
\begin{equation}
\vJ=\vj_c(\vr)\cos\omega t+\vj_s(\vr)\sin\omega t,
\label{eq:quad}
\end{equation}
where $\vj_c$ and $\vj_s$ are the two current quadratures. The same expansion gives, for $\alpha=c,s$,
\begin{equation}
\mathbf C_\alpha(\nhat)=\int \vj_\alpha(\vr')\,\cos(k\nhat\cdot\vr')\,\dd^3r',\qquad
\mathbf S_\alpha(\nhat)=\int \vj_\alpha(\vr')\,\sin(k\nhat\cdot\vr')\,\dd^3r',
\end{equation}
where volume and surface-current terms are included distributionally. In the complex representation of the standard multipole treatment~\cite{jackson}, the radiation amplitude is the Fourier integral of the current on the shell, and $\mathbf C_\alpha-i\,\mathbf S_\alpha=\int\vj_\alpha(\vr')\,e^{-ik\nhat\cdot\vr'}\dd^{3}r'$ recovers it for each quadrature. The real convention is kept throughout this paper. The radiation-zone vector potential is
\begin{equation}
\mathbf A_{\mathrm{rad}}=\frac{\mu_0}{4\pi r}\Bigl\{\cos(\omega t_r)\,[\mathbf C_c+\mathbf S_s]
+\sin(\omega t_r)\,[\mathbf C_s-\mathbf S_c]\Bigr\}.
\label{eq:Aradgen}
\end{equation}
An ideal compact source is called nonradiating when it produces no outgoing field in the source-free exterior under the retarded prescription. For a compact monochromatic outgoing solution, this is equivalent to a vanishing far-field pattern and zero total radiated power. A zero in only one observation direction is a node, not a nonradiating source. The far field is controlled by the transforms on the radiation shell $\vq=k\nhat$. A localized source is nonradiating precisely when the transverse parts of both bracketed combinations vanish for every direction $\nhat$~\cite{goedecke,devaneywolf,gbur}. For the separable even profiles used below, this condition reduces to the single transform $\JV(k\nhat)$.

\begin{proposition}[Uniqueness of the Nonradiating Maintainer]\label{prop:unique}
Fix two compactly supported spatial profiles of the stated distributional form whose currents maintain the same total charge density $\rhotot=0$ when each is driven by a common sinusoidal factor at frequency $\omega$. Suppose this maintenance holds for every $\omega>0$, and suppose both driven currents are nonradiating at every such frequency. Then the two profiles are identical. The complete-cancellation profile of Eq.~\eqref{eq:recipe} is therefore the unique maintainer in this class that never radiates.
\end{proposition}
\begin{proof}
The difference profile $\Delta\vj$ has vanishing distributional divergence, since both profiles evolve the same charge at each frequency. Its cosine and sine transforms are then purely transverse by the boundary-term argument above. Nonradiating at frequency $\omega$ requires the transverse parts of both transforms to vanish on the shell $q=\omega/c$ in every direction, so ranging over all $\omega>0$ makes the full Fourier transform of $\Delta\vj$ vanish for every $\vq\neq\mathbf 0$, and the value at $\vq=\mathbf 0$ follows by continuity. The transform of a compactly supported distribution extends to an entire function of the complex momentum by the Paley--Wiener--Schwartz theorem~\cite{hormander}, so vanishing for all real $\vq$ forces $\Delta\vj=\mathbf 0$.
\end{proof}

In vacuum, Eq.~\eqref{eq:vacuum} also holds for the radiative parts,
\begin{equation}
\vD_{\mathrm{rad}}=\varepsilon_0\vE_{\mathrm{rad}},\qquad
\vH_{\mathrm{rad}}=\frac{1}{\mu_0}\vB_{\mathrm{rad}} \qquad (\vr\in\Vm^c).
\end{equation}
Thus neither $\vD$-only nor $\vH$-only radiation can exist in vacuum.

\medskip
\noindent No outgoing vacuum solution can contain $\vD$ without $\vE$, or $\vH$ without $\vB$. The maintained magnetized sphere in Sec.~\ref{sec:bodies} gives the stronger case of complete four-current cancellation. It is neutral and satisfies $\rhotot=0$ and $\Jtot=\mathbf 0$ as distributions. Both retarded source terms in Eq.~\eqref{eq:wave} vanish, so the retarded prescription produces no source-generated $\vE$ or $\vB$ field and no radiation. The interior auxiliary values may still remain. For the magnetized sphere used here, $\vD=\mathbf 0$ and $\vH=-\vM$.

\section{Explicit Bodies and Symmetric Analogs}
\label{sec:bodies}
The compensated source conditions can be realized exactly in explicit bodies. Throughout this section, $\vP$ and $\vM$ are prescribed macroscopic source fields, as in remanent polarization or magnetization. Their constant amplitudes are $P_0$ in $\mathrm{C\,m^{-2}}$ and $M_0$ in $\mathrm{A\,m^{-1}}$. No induced linear law such as $\vP=\varepsilon_0\chi_e\vE$ or $\vM=\chi_m\vH$ is assumed. If $\vE=\mathbf 0$, the first law would force $\vP=\mathbf 0$. If $\vB=\mathbf 0$, the second gives $(1+\chi_m)\vM=\mathbf 0$ and therefore $\vM=\mathbf 0$ whenever $1+\chi_m\neq0$. The discussion begins with volume constructions that satisfy the propositions pointwise and then turns to two spheres whose cancellation occurs on the surface. Each sphere is a familiar polarized or magnetized body with one added free sheet that neutralizes its bound source.

\FloatBarrier
\subsection{Volume Constructions}
\label{sec:volume}
Both propositions can be realized pointwise in a compact body without any surface sheet. This is the special case of Eq.~\eqref{eq:recipe} in which the material field vanishes on the boundary. For the electric construction, choose the radial polarization
\begin{equation}
\vP(\vr)=\frac{P_0}{R}\Bigl(1-\frac{r}{R}\Bigr)\vr,\qquad r\le R,
\label{eq:Pvol}
\end{equation}
and set $\vP=\mathbf 0$ outside. Here $\vr$ is the position vector, so away from the origin $\vP=P_0(r/R)(1-r/R)\rhat$. The field is $C^1$ throughout the ball, including the origin, where $\partial_jP_i=(P_0/R)\bigl[(1-r/R)\delta_{ij}-x_ix_j/(Rr)\bigr]\to(P_0/R)\delta_{ij}$. Extending it by zero gives a continuous, piecewise $C^1$ field whose first derivative is discontinuous across $r=R$. Since $\vP$ vanishes at $r=R$, no bound surface charge appears. Its divergence is
\begin{equation}
\dive\vP=P_0\Bigl(\frac{3}{R}-\frac{4r}{R^2}\Bigr)
\label{eq:divPvol}
\end{equation}
which is nonzero in the bulk except on the sphere $r=3R/4$, where it changes sign as shown in Fig.~\ref{fig:volprofiles}. Set the free volume charge to $\rho_f=\dive\vP=-\rho_b$, with $\vM=\mathbf 0$ and no other free sources. The total free charge is zero because $\int\rho_f\dd V=\oint\vP\cdot\dd\mathbf a=0$. Thus $\rhotot=0$ everywhere without surface sheets. In electrostatics, $\curl\vE=\mathbf 0$, and the static assumption in Sec.~\ref{sec:maxwell} gives $\vE=\mathbf 0$ on all of $\mathbb{R}^3$. Therefore
\begin{equation}
\vD=\vP,\qquad \dive\vD=\dive\vP=\rho_f\not\equiv0 \quad (r<R),
\end{equation}
which realizes Proposition~\ref{prop:E} pointwise.

\begin{figure}[!htb]
\centering
\begin{tikzpicture}[line cap=round,x=3.5cm,y=0.72cm]
  \draw[semithick] (0,-1.4) rectangle (1,3.4);
  \foreach \x/\xl in {0/0.0,0.5/0.5,1/1.0}
    {\draw[semithick] (\x,-1.4) -- (\x,-1.22); \draw[semithick] (\x,3.4) -- (\x,3.22);
     \node[below] at (\x,-1.44) {\xl};}
  \foreach \x in {0.25,0.75}
    {\draw[semithick] (\x,-1.4) -- (\x,-1.29); \draw[semithick] (\x,3.4) -- (\x,3.29);}
  \foreach \y in {-1,0,1,2,3}
    {\draw[semithick] (0,\y) -- (0.028,\y); \draw[semithick] (1,\y) -- (0.972,\y);
     \node[left] at (-0.012,\y) {$\y$};}
  \draw[thin] (0,0) -- (1,0);
  \draw[thick,domain=0:1,smooth,samples=80] plot (\x,{3-4*\x});
  \draw[thick,dashed,domain=0:1,smooth,samples=80] plot (\x,{4*\x*(1-\x)});
  \draw[dotted,thick] (0.75,-1.4) -- (0.75,3.4);
  \node[above] at (0.6,3.42) {\footnotesize $r=3R/4$};
  \draw[thick] (0.08,-0.55) -- (0.2,-0.55); \node[right] at (0.21,-0.55) {\footnotesize $R\rho_f/P_0$};
  \draw[thick,dashed] (0.08,-1.05) -- (0.2,-1.05); \node[right] at (0.21,-1.05) {\footnotesize $4P_r/P_0$};
  \node[below=15pt] at (0.5,-1.4) {$r/R$};
  \node[anchor=north west] at (0.02,3.35) {(a)};
\end{tikzpicture}\hspace{0.55cm}
\begin{tikzpicture}[line cap=round,x=0.2692cm,y=3.456cm]
  \draw[semithick] (0,0) rectangle (13,1);
  \foreach \x in {0,2,4,6,8,10,12}
    {\draw[semithick] (\x,0) -- (\x,0.038); \draw[semithick] (\x,1) -- (\x,0.962);
     \node[below] at (\x,-0.012) {$\x$};}
  \foreach \y/\yl in {0/0.0,0.5/0.5,1/1.0}
    {\draw[semithick] (0,\y) -- (0.36,\y); \draw[semithick] (13,\y) -- (12.64,\y);
     \node[left] at (-0.16,\y) {$\yl$};}
  \draw[dotted,thick] (5.763459,0) -- (5.763459,1);
  \draw[dotted,thick] (9.095011,0) -- (9.095011,1);
  \draw[dotted,thick] (12.322941,0) -- (12.322941,1);
  \node[above] at (5.763,1.01) {\footnotesize $5.76$};
  \node[above] at (9.095,1.01) {\footnotesize $9.10$};
  \node[above] at (12.323,1.01) {\footnotesize $12.32$};
  \draw[thick] plot coordinates {
    (0.000000,0.000000) (0.100000,0.000005) (0.200000,0.000075) (0.300000,0.000378) (0.400000,0.001181)
    (0.500000,0.002848) (0.600000,0.005812) (0.700000,0.010568) (0.800000,0.017643) (0.900000,0.027575)
    (1.000000,0.040887) (1.100000,0.058062) (1.200000,0.079518) (1.300000,0.105585) (1.400000,0.136484)
    (1.500000,0.172308) (1.600000,0.213012) (1.700000,0.258400) (1.800000,0.308122) (1.900000,0.361677)
    (2.000000,0.418414) (2.100000,0.477548) (2.200000,0.538173) (2.300000,0.599284) (2.400000,0.659800)
    (2.500000,0.718591) (2.600000,0.774509) (2.700000,0.826417) (2.800000,0.873222) (2.900000,0.913904)
    (3.000000,0.947548) (3.100000,0.973368) (3.200000,0.990732) (3.300000,0.999181) (3.342093,1.000000)
    (3.400000,0.998446) (3.500000,0.988454) (3.600000,0.969333) (3.700000,0.941416) (3.800000,0.905224)
    (3.900000,0.861464) (4.000000,0.811004) (4.100000,0.754854) (4.200000,0.694139) (4.300000,0.630071)
    (4.400000,0.563918) (4.500000,0.496968) (4.600000,0.430500) (4.700000,0.365749) (4.800000,0.303875)
    (4.900000,0.245936) (5.000000,0.192863) (5.100000,0.145436) (5.200000,0.104269) (5.300000,0.069801)
    (5.400000,0.042282) (5.500000,0.021780) (5.600000,0.008181) (5.700000,0.001198) (5.763459,0.000000)
    (5.800000,0.000384) (5.900000,0.005154) (6.000000,0.014802) (6.100000,0.028530) (6.200000,0.045468)
    (6.300000,0.064710) (6.400000,0.085334) (6.500000,0.106435) (6.600000,0.127148) (6.700000,0.146674)
    (6.800000,0.164300) (6.900000,0.179417) (7.000000,0.191534) (7.100000,0.200288) (7.200000,0.205449)
    (7.300000,0.206923) (7.400000,0.204744) (7.500000,0.199075) (7.600000,0.190187) (7.700000,0.178454)
    (7.800000,0.164331) (7.900000,0.148335) (8.000000,0.131029) (8.100000,0.112998) (8.200000,0.094826)
    (8.300000,0.077082) (8.400000,0.060296) (8.500000,0.044947) (8.600000,0.031443) (8.700000,0.020114)
    (8.800000,0.011203) (8.900000,0.004859) (9.000000,0.001138) (9.095011,0.000000) (9.100000,0.000003)
    (9.200000,0.001329) (9.300000,0.004910) (9.400000,0.010472) (9.500000,0.017682) (9.600000,0.026161)
    (9.700000,0.035504) (9.800000,0.045292) (9.900000,0.055106) (10.000000,0.064544) (10.100000,0.073235)
    (10.200000,0.080848) (10.300000,0.087107) (10.400000,0.091795) (10.500000,0.094761) (10.600000,0.095924)
    (10.700000,0.095273) (10.800000,0.092862) (10.900000,0.088812) (11.000000,0.083298) (11.100000,0.076546)
    (11.200000,0.068817) (11.300000,0.060403) (11.400000,0.051612) (11.500000,0.042755) (11.600000,0.034135)
    (11.700000,0.026040) (11.800000,0.018728) (11.900000,0.012420) (12.000000,0.007295) (12.100000,0.003480)
    (12.200000,0.001052) (12.300000,0.000036) (12.322941,0.000000) (12.400000,0.000401) (12.500000,0.002069)
    (12.600000,0.004913) (12.700000,0.008770) (12.800000,0.013441) (12.900000,0.018708) (13.000000,0.024332)
    };
  \node[below=15pt] at (6.5,0) {$kR$};
  \node[rotate=90,above=25pt] at (0,0.5) {$F(kR)$};
  \node[anchor=north west] at (0.3,0.985) {(b)};
\end{tikzpicture}
\caption{(a) Scaled free charge $R\rho_f/P_0$ shown by the solid curve and polarization $4P_r/P_0$ shown by the dashed curve, from Eqs.~\eqref{eq:divPvol} and~\eqref{eq:Pvol}. The free charge changes sign at $r=3R/4$ and integrates to zero. (b) Normalized reduced radiation factor $F(kR)=j_2^{2}(kR)/\max_{x\ge0}j_2^{2}(x)$. The dotted lines mark the first three nonradiating roots. This panel isolates the finite-size Bessel factor in Eq.~\eqref{eq:power}. It is not the normalized total power in a frequency or size sweep because the prefactor changes as well.}
\label{fig:volprofiles}
\end{figure}
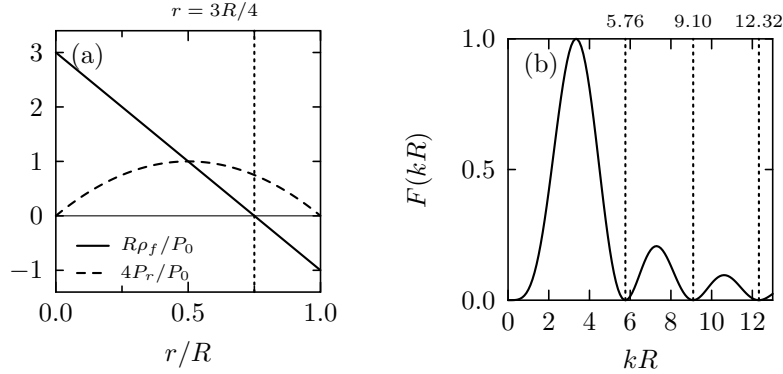
\FloatBarrier
\newpage

The magnetic construction uses a finite cylinder with azimuthal magnetization
\begin{equation}
\vM=M_0\,\frac{\rho}{R}\Bigl(1-\frac{\rho}{R}\Bigr)^{\!2}\Bigl(1-\frac{z^2}{L^2}\Bigr)^{\!2}\phihat,
\qquad \rho\le R,\ |z|\le L,
\label{eq:Mvol}
\end{equation}
and $\vM=\mathbf 0$ outside, where $\rho$ is the cylindrical radius. In Cartesian coordinates, $\vM=(M_0/R)(1-\rho/R)^{2}(1-z^{2}/L^{2})^{2}\,(-y,x,0)$. This form is well defined and $C^1$ on the axis, where $M_\phi/\rho\to(M_0/R)(1-z^{2}/L^{2})^{2}$. Since $\vM$ vanishes on the full boundary, no bound surface current appears. The volume current is
\begin{equation}
\Jb=\curl\vM=-\frac{\partial M_\phi}{\partial z}\,\rhohat+\frac{1}{\rho}\frac{\partial(\rho M_\phi)}{\partial \rho}\,\zhat,
\label{eq:curlM}
\end{equation}
with components
\begin{equation}
\begin{aligned}
(\curl\vM)_\rho
&=\frac{4M_0\,\rho z}{RL^2}\Bigl(1-\frac{\rho}{R}\Bigr)^{\!2}\Bigl(1-\frac{z^2}{L^2}\Bigr),
\\
(\curl\vM)_z
&=\frac{2M_0}{R}\Bigl(1-\frac{\rho}{R}\Bigr)\Bigl(1-\frac{2\rho}{R}\Bigr)\Bigl(1-\frac{z^2}{L^2}\Bigr)^{\!2}.
\end{aligned}
\label{eq:curlMcomp}
\end{equation}
The radial component vanishes on the midplane $z=0$, and the axial component vanishes on the cylinder $\rho=R/2$, but neither vanishes everywhere. The factor $1/\rho$ in Eq.~\eqref{eq:curlM} is regular because $(\curl\vM)_z\to(2M_0/R)(1-z^2/L^2)^2$ as $\rho\to0$. Set $\Jf=-\curl\vM$, with $\vP=\mathbf 0$ and $\rho_f=0$. This free current is divergence-free. Its normal component vanishes on the complete boundary because $(\curl\vM)_\rho\propto(1-\rho/R)^2$ on the side $\rho=R$ and $(\curl\vM)_z\propto(1-z^2/L^2)^2$ on the caps $z=\pm L$. The free charge is therefore conserved distributionally without surface terms. Since $\Jtot=\mathbf 0$ everywhere, magnetostatics and the static assumption give $\vB=\mathbf 0$. It follows that $\vH=-\vM$ and $\curl\vH=-\curl\vM=\Jf\not\equiv\mathbf 0$. This is a pointwise realization of Proposition~\ref{prop:B}.

\FloatBarrier
\subsection{Neutralized Electric Sphere}
\label{sec:electret}
Consider a sphere of radius $R$ with uniform polarization $\vP=P_0\,\zhat$, representing an electret. Uniform polarization gives no bound charge in the interior because $\rho_b=-\dive\vP=0$. The bound charge lies entirely on the surface, where $\sigma_b=\vP\cdot\nhat=P_0\cos\theta$. Coat the surface with the free charge $\sigma_f=-\sigma_b$, which is the charge part of Eq.~\eqref{eq:recipe}. The total charge then vanishes as a distribution, $\rhotot=0$. In electrostatics, the static assumption gives
\begin{equation}
\vE=\mathbf 0 \ \text{on all of }\mathbb{R}^3,\qquad
\vD=\varepsilon_0\vE+\vP=\vP\neq\mathbf 0 \ (r<R),
\end{equation}
with $\vD=\mathbf 0$ outside, as shown in Fig.~\ref{fig:electret}. This finite body therefore has $\vD=\vP$ and $\vE=\mathbf 0$. The neutralizing free layer carries the jump in the normal component of $\vD$ across $r=R$. Its effect is easiest to see by first removing it and solving the bare electret in full. Away from the sheet the total charge density vanishes, so the electrostatic potential obeys the Poisson equation $\nabla^{2}\Phi=-\rhotot/\varepsilon_0$ with zero right-hand side, which is the Laplace equation
\begin{equation}
\nabla^{2}\Phi=0 \qquad (r<R\ \text{and}\ r>R).
\label{eq:laplace}
\end{equation}
The problem is axisymmetric, so the separable solutions form the Legendre series
\begin{equation}
\Phi(r,\theta)=\sum_{l=0}^{\infty}\Bigl(A_l r^{l}+\frac{B_l}{r^{l+1}}\Bigr)P_l(\cos\theta).
\label{eq:legendre}
\end{equation}
Regularity at the origin removes every $B_l$ in the interior. Decay at infinity, which is the static assumption here, removes every $A_l$ in the exterior. The only source is the bound sheet $\sigma_b=P_0\cos\theta=P_0\,P_1(\cos\theta)$, whose angular content is pure $l=1$, so orthogonality of the Legendre polynomials removes every other multipole from the solution. The candidate is therefore
\begin{equation}
\Phi_{\mathrm{in}}=a\,r\cos\theta,\qquad \Phi_{\mathrm{out}}=\frac{b\cos\theta}{r^{2}}.
\label{eq:phicand}
\end{equation}
Two interface relations from Sec.~\ref{sec:bc} fix $a$ and $b$. The tangential component of $\vE$ is continuous, which for $\vE=-\nabla\Phi$ is continuity of $\Phi$ at $r=R$. The normal component of $\vD$ jumps by the free surface charge, Eq.~\eqref{eq:bcDE}, and this vanishes for the bare electret. Writing $D_r=\varepsilon_0E_r+P_r$, with $P_r=P_0\cos\theta$ inside and zero outside, converts $D_r^{\mathrm{out}}-D_r^{\mathrm{in}}=0$ into a jump condition on the field. The two specialized conditions are therefore
\begin{equation}
\Phi_{\mathrm{in}}(R,\theta)=\Phi_{\mathrm{out}}(R,\theta),\qquad
E_r^{\mathrm{out}}(R,\theta)-E_r^{\mathrm{in}}(R,\theta)=\frac{P_0\cos\theta}{\varepsilon_0},
\label{eq:bcspecE}
\end{equation}
the first from the tangential condition of Eq.~\eqref{eq:bcEB} and the second from the normal condition of Eq.~\eqref{eq:bcDE} with $\sigma_f=0$. With $E_r=-\partial\Phi/\partial r$, so that $E_r^{\mathrm{in}}=-a\cos\theta$ and $E_r^{\mathrm{out}}=2b\cos\theta/r^{3}$, Eq.~\eqref{eq:bcspecE} becomes
\begin{equation}
aR=\frac{b}{R^{2}},\qquad \frac{2b}{R^{3}}+a=\frac{P_0}{\varepsilon_0}.
\label{eq:matching}
\end{equation}
The first gives $b=aR^{3}$, and substitution into the second gives $3a=P_0/\varepsilon_0$, so
\begin{equation}
a=\frac{P_0}{3\varepsilon_0},\qquad b=\frac{P_0R^{3}}{3\varepsilon_0}.
\label{eq:absol}
\end{equation}
Since $\nabla(r\cos\theta)=\nabla z=\zhat$, the interior field is uniform. The bare electret has $\vE_{\mathrm{in}}=-a\,\zhat=-\vP/(3\varepsilon_0)$ and $\vD_{\mathrm{in}}=\varepsilon_0\vE_{\mathrm{in}}+\vP=\tfrac{2}{3}\vP$, together with an exterior dipole field. Its moment is $4\pi\varepsilon_0\,b\,\zhat=\tfrac{4}{3}\pi R^{3}P_0\,\zhat$, which is exactly the bound dipole $\int_V\vP\dd V$. The free charge sheet removes both the interior depolarizing field and the exterior electric dipole.

\begin{figure}[!htb]
\centering
\begin{tikzpicture}[>=stealth,line cap=round]
  \draw[very thick] (0,0) circle (1.5);
  \foreach \x in {-0.65,0,0.65} \draw[->,thick] (\x,-0.8) -- (\x,0.8);
  \node at (1.05,0.55) {$\vP$};
  \draw (0.42,1.44) -- (1.15,2.05);
  \node[right] at (1.15,2.05) {$\sigma_f=-P_0\cos\theta$ (Free)};
  \draw (-0.42,1.44) -- (-1.15,2.05);
  \node[left] at (-1.15,2.05) {$\sigma_b=+P_0\cos\theta$ (Bound)};
  \node at (0,2.42) {$\sigma_b>0,\ \sigma_f<0$};
  \draw (0,-1.5) -- (0,-1.95);
  \node[below] at (0,-1.95) {$\sigma_b<0,\ \sigma_f>0$ (Lower Hemisphere)};
  \node[align=left,anchor=west] at (2.7,-0.15)
    {$\vE=\mathbf{0}$ everywhere\\[2pt]
     $\vD=\vP$ \ ($r<R$)\\[2pt]
     $\vD=\mathbf{0}$ \ ($r>R$)};
\end{tikzpicture}
\caption{Neutralized electret with uniform $\vP=P_0\zhat$. The bound sheet $\sigma_b=\vP\cdot\nhat$ and the free sheet $\sigma_f=-\sigma_b$ coincide on $r=R$. Their signs reverse between the two hemispheres, and the resulting fields are shown.}
\label{fig:electret}
\end{figure}
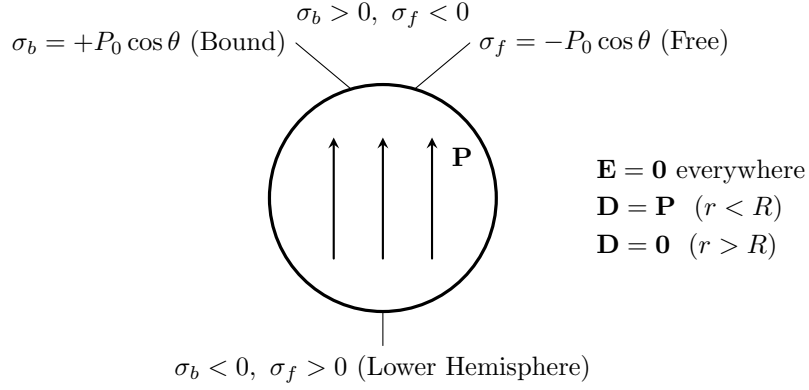
\FloatBarrier

For the time-dependent problem, let the polarization oscillate as $\vP(t)=P_s(t)\,\zhat$, where $P_s(t)=P_0\cos\omega t$. Keep the neutralizing layer matched at every instant, so $\sigma_f(t)=-P_s(t)\cos\theta$. A free current is needed to maintain this sheet. Any realization must satisfy $\partial\sigma_f/\partial t+\divs\vK_f=(\Jf^{\mathrm{in}}-\Jf^{\mathrm{out}})\cdot\nhat$. Two choices show the difference between charge cancellation and complete current cancellation. In both cases, the prescribed current represents electrodes and a driving circuit rather than the free response of an isolated electret. The first choice is the volume-current realization, for which $\vK_f=\mathbf 0$. Set $\Jf=-\partial\vP/\partial t$ inside the sphere, as prescribed by the current part of Eq.~\eqref{eq:recipe}. The surface-continuity law reduces to $\Jf\cdot\nhat|_{r=R^-}=-(\dd P_s/\dd t)\cos\theta=\partial\sigma_f/\partial t$, matching the way the polarization current feeds the bound sheet. The volume currents cancel, the coincident charge sheets cancel, and $\rhotot=0$ with $\Jtot=\mathbf 0$. Both retarded source terms in Eq.~\eqref{eq:wave} vanish. Under the retarded prescription, this source does not radiate, just like the compensated magnetized sphere below.
\\

The second choice is the tangential-sheet realization, with no free volume current. Maintain $\sigma_f$ using the tangential sheet on $r=R$,
\begin{equation}
\vK_f=\tfrac{1}{2}R\,\frac{\dd P_s}{\dd t}\,\sin\theta\,\thetahat.
\label{eq:Kf}
\end{equation}
Its surface divergence satisfies the required continuity law. Since $\sigma_f=-P_s(t)\cos\theta$,
\begin{equation}
\divs\vK_f=\frac{1}{R\sin\theta}\frac{\partial}{\partial\theta}
\Bigl[\sin\theta\cdot\tfrac{1}{2}R\,\frac{\dd P_s}{\dd t}\sin\theta\Bigr]=\frac{\dd P_s}{\dd t}\cos\theta=-\frac{\partial\sigma_f}{\partial t}.
\label{eq:surfcont}
\end{equation}
The total current is now nonzero, although global charge cancellation keeps it divergence-free through Eq.~\eqref{eq:divJ}. Since $\rhotot=0$ as a distribution at every instant, every multipole moment of the total charge density vanishes. The free sheet exactly cancels the bound dipole $\vp_b(t)=\int_V\vP\dd V=\tfrac{4}{3}\pi R^3P_s(t)\,\zhat$,
\begin{equation}
\vp_f(t)=\oint_S\vr\,\sigma_f\dd S=-P_s(t)R^3\!\int\rhat\cos\theta\dd\Omega
=-\tfrac{4}{3}\pi R^3P_s(t)\,\zhat=-\vp_b(t),
\end{equation}
because only the $z$ component survives and $\int\cos^2\theta\dd\Omega=4\pi/3$. The ordinary long-wavelength electric-dipole channel is therefore closed. The magnetic-dipole moment also vanishes. The volume integrand $\vr\times\zhat$ has Cartesian components $(y,-x,0)$ and is purely azimuthal. The sheet integrand $\sin\theta\,\nhat\times\thetahat=\sin\theta\,\phihat$ is also azimuthal. Both angular integrals vanish, giving $\mathbf m=\tfrac12\int\vr\times\Jtot\,\dd^3r=\mathbf 0$. Any remaining radiation is therefore not an ordinary magnetic-dipole contribution. Finite-size structure still remains in the divergence-free current, so radiation is not excluded. Its radiation-zone contribution can be found without a long-wavelength approximation. For the tangential-sheet construction, write $\Jtot(\vr,t)=g(t)\,\vj(\vr)$ using the distributional profile in Eq.~\eqref{eq:transform}, with $\vj_V=\zhat$, $\vk=\tfrac12R\sin\theta\,\thetahat$, and $\delta_S=\delta(r-R)$. The profile transform is
\begin{equation}
\JV(\vq)=2\pi R^{3}\,j_2(qR)\,\bigl[\zhat-(\zhat\cdot\qhat)\qhat\bigr],
\label{eq:Jq}
\end{equation}
where $j_2$ is the spherical Bessel function. Appendix~\ref{app:jq} gives the derivation. The transform is purely transverse, as required by Eq.~\eqref{eq:divJ}, and it is even in $\vq$, so only the cosine transform survives. Equation~\eqref{eq:Jq} is written for $\vq\neq\mathbf 0$, while $\JV(\mathbf 0)=\mathbf 0$ follows by continuity. Substituting Eq.~\eqref{eq:Jq} into Eq.~\eqref{eq:onshell}, using $g'(t)=\dd^{2}P_s/\dd t^{2}=-\omega^{2}P_0\cos\omega t$ and $\zhat-(\zhat\cdot\nhat)\nhat=-\sin\theta\,\thetahat$, gives the radiation-zone field. Here $\theta$ is measured from $\zhat$. The result is exact for arbitrary $kR$ at leading order in $1/r$,
\begin{equation}
\vE_{\mathrm{rad}}(\vr,t)=-\frac{\mu_0\,\omega^{2}R^{3}P_0}{2r}\,j_2(kR)\,
\cos[\omega(t-r/c)]\,\sin\theta\,\thetahat,
\label{eq:Erad}
\end{equation}
where $\vB_{\mathrm{rad}}=\nhat\times\vE_{\mathrm{rad}}/c$. The time-averaged Poynting flux is $\langle\mathbf S\rangle=\langle\vE_{\mathrm{rad}}^{\,2}\rangle\,\nhat/(\mu_0c)$. Using $\langle\cos^{2}\rangle=\tfrac12$ and $\int\sin^{2}\theta\dd\Omega=8\pi/3$ gives the total radiated power
\begin{equation}
\overline{\mathcal P}=\frac{\pi\mu_0\,\omega^{4}R^{6}P_0^{2}}{3c}\,j_2^{2}(kR).
\label{eq:power}
\end{equation}
Figure~\ref{fig:volprofiles}(b) shows the normalized factor $F(kR)=j_2^{2}(kR)/\max j_2^{2}$. In a frequency or size sweep, the prefactor in Eq.~\eqref{eq:power} changes as well. The field pattern is proportional to $\sin\theta$, and the power pattern is proportional to $\sin^{2}\theta$. Since $Y_{10}=\sqrt{3/4\pi}\,\cos\theta$ and $\sqrt{4\pi/3}\,(\partial Y_{10}/\partial\theta)\,\thetahat=-\sin\theta\,\thetahat$, the field has the electric-type $l=1$ vector harmonic. Here that harmonic comes from the finite extent of a divergence-free current rather than from a charge dipole. In a Cartesian long-wavelength expansion, the same contribution can be represented through toroidal and higher mean-square-radius terms in the electric-parity coefficient~\cite{dubovik}. Whether toroidal terms should be treated as an independent radiation family depends on convention~\cite{fernandez}. The convention-independent result is the exact on-shell factor $j_2(kR)$. Earlier work on time-dependent toroidal sources found finite radiationless counterparts~\cite{afanasiev}. The present construction starts instead with a compensated body and produces an explicit finite-size coefficient. Abbott and Griffiths treated related radiationless acceleration in infinite cylindrical and planar current distributions~\cite{abbottgriffiths}. The sphere provides a compact compensated realization with $\rhotot\equiv0$.
\\

The exterior field admits a closed form. The pattern of Eq.~\eqref{eq:Erad} is a pure electric-type $l=1$ multipole, and this holds at every $kR$, not only in the long-wavelength limit. An outgoing exterior solution is determined by its far-field pattern, by the same uniqueness theorem invoked below. The complete exterior field of the compensated sphere therefore coincides, for every $r>R$, with the field of a point electric dipole at the origin whose moment is the effective, frequency-dependent quantity
\begin{equation}
\vp_{\mathrm{eff}}(k)=2\pi P_0R^{3}\,j_2(kR)\,\zhat.
\label{eq:peff}
\end{equation}
The full fields follow from the standard oscillating-dipole solution of the multipole treatment~\cite{jackson}, written here in real form with the retarded phase carried by $t_r\equiv t-r/c$ and with the signed scalar amplitude $p_{\mathrm{eff}}(k)=2\pi P_0R^{3}j_2(kR)$, so that $\vp_{\mathrm{eff}}=p_{\mathrm{eff}}\,\zhat$,
\begin{equation}
\scalebox{0.94}{%
$\begin{aligned}
\vE(\vr,t)
&=\frac{p_{\mathrm{eff}}}{4\pi\varepsilon_0}
\Bigl\{-\frac{k^{2}\sin\theta\,\cos(\omega t_r)}{r}\,\thetahat
+\bigl(2\cos\theta\,\rhat+\sin\theta\,\thetahat\bigr)
\Bigl[\frac{\cos(\omega t_r)}{r^{3}}
-\frac{k\sin(\omega t_r)}{r^{2}}\Bigr]\Bigr\},
\\
\vB(\vr,t)
&=-\frac{\mu_0\,\omega k\,p_{\mathrm{eff}}}{4\pi r}\,
\sin\theta\,
\Bigl[\cos(\omega t_r)+\frac{\sin(\omega t_r)}{kr}\Bigr]\,\phihat,
\qquad r>R.
\end{aligned}$%
}
\label{eq:exterior}
\end{equation}
These fields satisfy the source-free Maxwell equations for every $r>0$ and reduce to Eq.~\eqref{eq:Erad} at leading order in $1/r$. The remaining terms give the exact intermediate and near zones outside the source. As $kR\to0$, $\vp_{\mathrm{eff}}\to(2\pi/15)P_0R^{5}k^{2}\,\zhat$, which vanishes quasistatically and matches the $\omega^{8}R^{10}$ power law below. The textbook dipole formula supplies an independent check, since $\overline{\mathcal P}=\mu_0\omega^{4}p_{\mathrm{eff}}^{2}/(12\pi c)$ applied to Eq.~\eqref{eq:peff} reproduces Eq.~\eqref{eq:power} exactly. The boundary-value solution of the companion paper gives the same moment as $\tfrac{3}{2}j_2(kR)$ times the bare dipole $\tfrac{4}{3}\pi R^{3}P_0$, which equals Eq.~\eqref{eq:peff} identically, so the on-shell transform and the direct matching agree at every $kR$~\cite{twocurrents}.
\\

Because the ordinary charge-dipole term vanishes, the first nonzero contribution is suppressed. For $kR\ll1$, $j_2(kR)\simeq(kR)^2/15$, and Eq.~\eqref{eq:power} gives $\overline{\mathcal P}\propto\omega^{8}R^{10}$. The amplitude vanishes on the axis $\theta=0,\pi$ at every frequency, but these are only directional nodes. At a nonzero root of $j_2(kR)$, $\JV(k\nhat)=\mathbf 0$ for every direction. Outside the compact source, the fields then satisfy the homogeneous Helmholtz equations and the outgoing condition. The standard exterior uniqueness theorem indicates that a vanishing far-field pattern forces the outgoing exterior field to vanish, not only its $1/r$ part~\cite{devaneywolf}. These are the exact nonradiating frequencies of the idealized source, shown in Fig.~\ref{fig:volprofiles}(b). Equation~\eqref{eq:exterior} makes the same statement, since $\vp_{\mathrm{eff}}$ vanishes at each root and the entire exterior field, near zone included, vanishes with it. The companion solution reaches the same conclusion from a single exterior coefficient and needs no uniqueness theorem~\cite{twocurrents}. The first occurs at $kR=5.763\ldots$. These frequencies lie well outside the long-wavelength area. At the first root the sphere diameter equals $1.835$ wavelengths, and the later roots occur at still larger size-to-wavelength ratios, so the vanishing is a property of the full finite-size coefficient rather than of any truncated multipole order. The nulls are quadratically sharp. Near a root $x_n$ of $j_2$, Eq.~\eqref{eq:peff} gives $p_{\mathrm{eff}}\approx2\pi P_0R^{3}j_2'(x_n)\,(kR-x_n)$, so the power of Eq.~\eqref{eq:power} rises as $(kR-x_n)^{2}$ from each null, with $j_2'(x_n)=-0.166$, $0.108$, and $-0.080$ at the first three roots. Nonradiating states of the nonzero currents are studied as dynamic anapoles, realized in metamaterials and dielectric nanoparticles~\cite{papasimakis,miroshnichenko}. The compensated sphere reaches its null frequencies through the single exact form factor of Eq.~\eqref{eq:peff} rather than through interference between separate multipole families. The current remains nontrivial, so the electric and magnetic fields cannot both vanish throughout the source region. Otherwise the Amp\`ere-Maxwell equation would require $\Jtot=\mathbf 0$.
\\

Whenever radiation remains, Eq.~\eqref{eq:onshell} shows it to be an ordinary transverse wave carried by $\vE$ and $\vB$, with $\nhat\cdot\vE_{\mathrm{rad}}=\nhat\cdot\vB_{\mathrm{rad}}=0$ and $\vB_{\mathrm{rad}}=\nhat\times\vE_{\mathrm{rad}}/c$, and Eq.~\eqref{eq:vacuum} then fixes $\vD_{\mathrm{rad}}$ and $\vH_{\mathrm{rad}}$ as in Sec.~\ref{sec:radiation}.

\FloatBarrier
\subsection{Compensated Magnetized Sphere}
\label{sec:magnet}
The magnetic analog is a uniformly magnetized sphere of radius $R$, with $\vM=M_0\,\zhat$ and $\vP=\mathbf 0$. Uniform magnetization gives $\Jb=\curl\vM=\mathbf 0$ in the volume. The bound current lies entirely on the surface sheet $\vK_b=\vM\times\nhat=M_0\sin\theta\,\phihat$. For the bare sphere, the calculation parallels the electret of Eqs.~\eqref{eq:laplace} through~\eqref{eq:absol} step for step and shows the magnetic scalar potential at work. There is no free current, so $\curl\vH=\mathbf 0$ in each region and $\vH=-\nabla\Phi_m$ there. Combining $\dive\vB=0$ with $\vB=\mu_0(\vH+\vM)$ gives $\dive\vH=-\dive\vM$, and the uniform interior magnetization is divergence-free in the bulk. The potential therefore obeys
\begin{equation}
\nabla^{2}\Phi_m=0 \qquad (r<R\ \text{and}\ r>R),
\label{eq:laplaceM}
\end{equation}
driven only by the effective magnetic surface charge $\sigma_m=\vM\cdot\nhat=M_0\cos\theta$ on $r=R$. The Legendre series of Eq.~\eqref{eq:legendre}, regularity at the origin, decay at infinity, and the pure $P_1(\cos\theta)$ content of $\sigma_m$ select the same $l=1$ candidate as before,
\begin{equation}
\Phi_{m,\mathrm{in}}=a\,r\cos\theta,\qquad \Phi_{m,\mathrm{out}}=\frac{b\cos\theta}{r^{2}}.
\label{eq:phicandM}
\end{equation}
The interface relations of Sec.~\ref{sec:bc} supply two conditions. The tangential component of $\vH$ is continuous because $\vK_f=\mathbf 0$ in Eq.~\eqref{eq:bcHB}, which is continuity of $\Phi_m$. The normal component of $\vB$ is continuous by Eq.~\eqref{eq:bcEB}, and writing $B_r=\mu_0(H_r+M_r)$ with $M_r=M_0\cos\theta$ inside converts that statement into a jump condition on $H_r$. The specialized pair is
\begin{equation}
\Phi_{m,\mathrm{in}}(R,\theta)=\Phi_{m,\mathrm{out}}(R,\theta),\qquad
H_r^{\mathrm{out}}(R,\theta)-H_r^{\mathrm{in}}(R,\theta)=M_0\cos\theta.
\label{eq:bcspecM}
\end{equation}
The components follow from the candidate potentials,
\begin{equation}
H_r^{\mathrm{in}}=-a\cos\theta,\quad
H_r^{\mathrm{out}}=\frac{2b\cos\theta}{r^{3}},\quad
H_\theta^{\mathrm{in}}=a\sin\theta,\quad
H_\theta^{\mathrm{out}}=\frac{b\sin\theta}{r^{3}},
\label{eq:Hcomp}
\end{equation}
and substituting them into Eq.~\eqref{eq:bcspecM} gives the matching equations
\begin{equation}
aR=\frac{b}{R^{2}},\qquad \frac{2b}{R^{3}}+a=M_0.
\end{equation}
The first relation expresses continuity of $H_\theta$, the second continuity of $B_r$. The first gives $b=aR^{3}$, substitution into the second gives $3a=M_0$, and so
\begin{equation}
a=\frac{M_0}{3},\qquad b=\frac{M_0R^{3}}{3}.
\label{eq:absolM}
\end{equation}
Since $\nabla(r\cos\theta)=\zhat$, the interior field is uniform,
\begin{equation}
\vH_{\mathrm{in}}=-\tfrac13\vM,\qquad
\vB_{\mathrm{in}}=\mu_0(\vH_{\mathrm{in}}+\vM)=\tfrac23\mu_0\vM,
\label{eq:magfields}
\end{equation}
together with an exterior dipole field of moment $4\pi b\,\zhat=\tfrac{4}{3}\pi R^{3}M_0\,\zhat$, which is exactly $\int_V\vM\dd V$. The bound sheet appears in the remaining tangential condition of Eq.~\eqref{eq:bcHB}. With $b=aR^{3}$ and the components of Eq.~\eqref{eq:Hcomp},
\begin{equation}
\frac{1}{\mu_0}\,\nhat\times(\vB_{\mathrm{out}}-\vB_{\mathrm{in}})\Big|_{r=R}
=\Bigl(\frac{b}{R^{3}}-a+M_0\Bigr)\sin\theta\,\phihat
=M_0\sin\theta\,\phihat=\vM\times\nhat=\vK_b,
\label{eq:Kbcheck}
\end{equation}
as Eq.~\eqref{eq:bcHB} requires with $\vK_f=\mathbf 0$. Now add the free surface current $\vK_f=-\vK_b$ from Eq.~\eqref{eq:recipe}. The total surface current vanishes, and $\Jb=\mathbf 0$ in the volume, so $\Jtot=\mathbf 0$ distributionally on all of $\mathbb{R}^3$. In magnetostatics, $\dive\vB=0$ and $\curl\vB=\mathbf 0$ everywhere. The static assumption and uniqueness therefore give
\begin{equation}
\vB=\mathbf 0\ \text{everywhere},\qquad
\vH=\frac{1}{\mu_0}\vB-\vM=-\vM\neq\mathbf 0\ (r<R),
\end{equation}
with $\vH=\mathbf 0$ outside, as shown in Fig.~\ref{fig:magnet}. Relative to the bare sphere, the free current sheet removes both the interior field $\vB=\tfrac23\mu_0\vM$ and the exterior dipole field. Since $\vM$ is uniform, the bulk relation is $\curl\vH=\mathbf 0=\Jf$. This is the compact surface-current counterpart of Proposition~\ref{prop:B}.
\\

For a time-dependent source construction, prescribe $\vM(t)=M_0\cos\omega t\,\zhat$ and maintain $\vK_f(t)=-\vK_b(t)$. The azimuthal sheets are surface-divergence-free, so an initially neutral source develops no surface charge. Thus $\rhotot=0$ and $\Jtot=\mathbf 0$ for all $t$. Both wave-equation sources in Eq.~\eqref{eq:wave} vanish, and the retarded prescription gives no radiation.

\begin{figure}[!htb]
\centering
\begin{tikzpicture}[>=stealth,line cap=round]
  \draw[very thick] (0,0) circle (1.5);
  \draw[->,very thick] (0,-0.8) -- (0,0.8);
  \node at (0.35,0.55) {$\vM$};
  \node[inner sep=1pt] at (13:1.74) {\large$\odot$};
  \node[inner sep=1pt] at (13:1.26) {\large$\otimes$};
  \node[inner sep=1pt] at (167:1.74) {\large$\otimes$};
  \node[inner sep=1pt] at (167:1.26) {\large$\odot$};
  \draw (1.87,0.44) -- (2.42,1.02);
  \node[right] at (2.45,1.05) {$\vK_b=\vM\times\nhat$ (Bound)};
  \draw (1.39,0.18) -- (1.80,0.13);
  \node[right] at (1.85,0.12) {$\vK_f=-\vK_b$ (Free)};
  \node[align=left,anchor=north west] at (1.85,-0.42)
    {$\vB=\mathbf{0}$ everywhere\\[2pt]
     $\vH=-\vM$ \ ($r<R$)\\[2pt]
     $\vH=\mathbf{0}$ \ ($r>R$)};
\end{tikzpicture}
\caption{Compensated magnetized sphere with uniform $\vM=M_0\zhat$ and $\vP=\mathbf 0$. The coincident bound and free current sheets on $r=R$ are opposite at every point. The symbols $\odot$ and $\otimes$ indicate currents out of and into the page. They are drawn with a small radial offset only for visibility. The resulting fields are shown.}
\label{fig:magnet}
\end{figure}
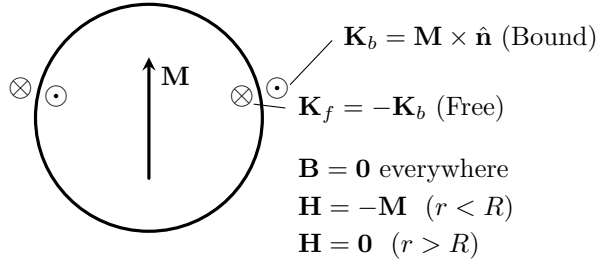
\FloatBarrier

The sphere constructions require surface sources that are maintained in time. They are macroscopic idealizations of electrodes or a wound coil, and the coincident free and bound sheets are distributional limits. A real coating or winding would stand a finite distance from the material and would be connected by leads. It could only approximate the cancellation, and its supply circuit would be part of the full current distribution. The radiation statements therefore apply to the prescribed ideal sources. For a real device, the ``vacuum exterior'' must lie outside every source-bearing component, including electrodes, leads, and return currents.

\FloatBarrier
\subsection{Infinite Planar and Cylindrical Analogs}
\label{sec:analogs}
The distributional compensation rule in Eq.~\eqref{eq:recipe} does not depend on geometry. Infinite planar and cylindrical examples also show the reverse pattern, in which the auxiliary field vanishes while the primary field remains. A uniformly polarized slab occupying $-d/2<z<d/2$, with $\vP=P_0\zhat$, carries only the face charges $\sigma_b=\pm P_0$. A single infinite sheet of uniform charge density $\sigma$ has, by translation symmetry, a field $E_z(z)\,\zhat$ that is odd about the sheet, and the pillbox condition of Eq.~\eqref{eq:bcDE} with the no-background choice fixes it,
\begin{equation}
E_z=\pm\frac{\sigma}{2\varepsilon_0}\qquad(\text{above and below the sheet}).
\label{eq:sheetfield}
\end{equation}
Between the faces, the sheet at $z=+d/2$ with $\sigma_b=+P_0$ contributes $-(P_0/2\varepsilon_0)\zhat$ and the sheet at $z=-d/2$ with $\sigma_b=-P_0$ contributes the same, while outside the pair the two contributions cancel. Superposition therefore gives
\begin{equation}
\vE_{\mathrm{in}}=-\frac{P_0}{\varepsilon_0}\,\zhat,\qquad
\vE_{\mathrm{out}}=\mathbf 0,\qquad
\vD=\varepsilon_0\vE+\vP=\mathbf 0\ \text{everywhere}.
\label{eq:slabfields}
\end{equation} This is the reverse of the neutralized sphere because $\vE$ survives while $\vD$ vanishes. Coating both faces with $\sigma_f=-\sigma_b$ restores the sphere pattern, with $\vE=\mathbf 0$ everywhere and $\vD=\vP$ inside. The magnetic counterpart is an infinite cylinder magnetized along its axis, with $\vM=M_0\zhat$ and $\vP=\mathbf 0$. Its bound sheet $\vK_b=M_0\phihat$ is solenoidal. Translation symmetry, azimuthal symmetry, Maxwell's equations, and regularity permit only constant axial fields inside and outside. The radial and azimuthal components are excluded by $\dive\vB=0$ with regularity on the axis and by circular Amp\`ere loops that enclose no axial current, and $\curl\vB=\mathbf 0$ in each region then makes the surviving $B_z$ constant there. Choose no imposed uniform exterior background, so $\vB_{\mathrm{out}}=\mathbf 0$. The tangential condition of Eq.~\eqref{eq:bcHB} at $\rho=R$, with $\vK_f=\mathbf 0$ and $\rhohat\times\zhat=-\phihat$, reads
\begin{equation}
\frac{1}{\mu_0}\,\rhohat\times(\vB_{\mathrm{out}}-\vB_{\mathrm{in}})
=\frac{B_{\mathrm{in}}}{\mu_0}\,\phihat=\vK_b=M_0\,\phihat,
\label{eq:cyljump}
\end{equation}
\\
so
\\
\begin{equation}
\vB_{\mathrm{in}}=\mu_0M_0\,\zhat=\mu_0\vM,\qquad
\vH=\frac{\vB}{\mu_0}-\vM=\mathbf 0\ \text{everywhere}.
\label{eq:cylfields}
\end{equation} Adding $\vK_f=-\vK_b$ gives $\vB=\mathbf 0$ everywhere and $\vH=-\vM$ inside. In this infinite geometry, the zero exterior constant is a boundary choice rather than a consequence of decay at infinity.
\\

Together, the four bodies show every static possibility. Either field in the pair $(\vE,\vD)$ or $(\vB,\vH)$ can vanish inside matter while the other remains. The complete source distribution and boundary data determine which field vanishes. Equation~\eqref{eq:constit} then fixes the relation between the surviving auxiliary and primary fields through $\vP$ or $\vM$. The infinite examples are symmetry limits and do not contradict the uniqueness results for compact localized sources.

\FloatBarrier
\section{Boundary Conditions}
\label{sec:bc}
The same issue appears at a boundary between matter and vacuum, and the interface relations used throughout follow from the integral forms of the macroscopic equations. Consider a flat patch of the surface with unit normal $\nhat$ pointing from the material into the vacuum. For the normal components, integrate each divergence equation over a small pillbox that straddles the patch, with faces of area $\Delta A$ parallel to the surface and vanishing height. The volume integral of a divergence equals the outward flux, the side walls contribute nothing in the vanishing-height limit, and the enclosed source reduces to the surface density times $\Delta A$. Applied to $\dive\vD=\rho_f$, to $\dive\vE=\rhotot/\varepsilon_0$, and to $\dive\vB=0$, the pillbox gives the normal conditions below, with enclosed surface charges $\sigma_f$, $\sigma_f+\sigma_b$, and zero. For the tangential components, integrate each curl equation over a small rectangular loop with long sides of length $\Delta\ell$ parallel to the surface on either side and vanishing short sides. The circulation picks out the tangential components, the fluxes of $\partial\vB/\partial t$ and $\partial\vD/\partial t$ vanish with the loop area for bounded fields, and the enclosed sheet current reduces to the surface density times $\Delta\ell$. Applied to $\curl\vE=-\partial\vB/\partial t$, to $\curl\vH=\Jf+\partial\vD/\partial t$, and to the $\vB$ form of the Amp\`ere-Maxwell law, the loop gives the tangential conditions. The results are
\begin{align}
\nhat\cdot(\vD_{\mathrm{vac}}-\vD_{\mathrm{mat}})&=\sigma_f, &
\nhat\cdot(\vE_{\mathrm{vac}}-\vE_{\mathrm{mat}})&=\frac{\sigma_f+\sigma_b}{\varepsilon_0},
\label{eq:bcDE}\\
\nhat\times(\vH_{\mathrm{vac}}-\vH_{\mathrm{mat}})&=\vK_f, &
\frac{1}{\mu_0}\nhat\times(\vB_{\mathrm{vac}}-\vB_{\mathrm{mat}})&=\vK_f+\vK_b.
\label{eq:bcHB}
\end{align}
The remaining two conditions at an ordinary interface carry no surface source at all,
\begin{equation}
\nhat\times(\vE_{\mathrm{vac}}-\vE_{\mathrm{mat}})=\mathbf 0,\qquad
\nhat\cdot(\vB_{\mathrm{vac}}-\vB_{\mathrm{mat}})=0.
\label{eq:bcEB}
\end{equation} If the surface sources cancel, so that $\sigma_f+\sigma_b=0$ and $\vK_f+\vK_b=\mathbf 0$, then the normal component of $\vE$ and the tangential component of $\vB$ are continuous. The normal component of $\vD$ still jumps by $\sigma_f$, and the tangential component of $\vH$ jumps by $\vK_f$. Equation~\eqref{eq:constit} resolves the apparent mismatch. Substituting it into the auxiliary boundary conditions reproduces them exactly. The difference between the primary and auxiliary jumps is the polarization surface charge $\sigma_b=\vP\cdot\nhat$ or the magnetization surface current $\vK_b=\vM\times\nhat$ already contained in the definitions.

\section{Structural Comparison of the Electric and Magnetic Cases}
Table~\ref{tab:structure} collects the static correspondence between the electric and magnetic constructions. The dynamic cases differ because charge cancellation gives $\rhotot=0$ and $\dive\Jtot=0$ but may leave $\Jtot\neq\mathbf 0$. Complete four-current cancellation also sets $\Jtot=\mathbf 0$ and removes the retarded source-generated fields.

\begin{table}[!ht]
\caption{Static correspondence between the electric and magnetic constructions. The compact-body rows assume the sources in Secs.~\ref{sec:volume} through~\ref{sec:magnet} and the static assumption. The scalar-potential statements are local. A global single-valued potential also requires zero circulation around every noncontractible loop. The potential $\Psi_B$ is used only where $\curl\vB=\mathbf 0$ and is distinct from the $\vH$ potential $\Phi_m$ of Sec.~\ref{sec:magnet}.}
\label{tab:structure}
\centering
\small
\begin{tabularx}{\textwidth}{@{}YY@{}}
\toprule
\textbf{Electric Case} & \textbf{Magnetic Case ($\vP=\mathbf 0$)} \\
\midrule
\multicolumn{2}{@{}l}{\emph{Cancellation}}\\
$\rho_f+\rho_b=0$ & $\Jf+\curl\vM=\mathbf 0$ \\
\addlinespace[2pt]
\multicolumn{2}{@{}l}{\emph{Definitions and Macroscopic Laws}}\\
$\vD=\varepsilon_0\vE+\vP$ & $\vH=\vB/\mu_0-\vM$ \\
$\dive\vD=\rho_f$ & $\curl\vH=\Jf+\partial\vD/\partial t$ \\
$\dive\vE=\rhotot/\varepsilon_0$ & $\curl\vB=\mu_0\Jtot+c^{-2}\partial\vE/\partial t$ \\
Vacuum with $\vP=\vM=\mathbf 0$ gives $\vD=\varepsilon_0\vE$ & Vacuum gives $\vH=\vB/\mu_0$ \\
\addlinespace[2pt]
\multicolumn{2}{@{}l}{\emph{Local Static Cancellation Inside $\Vm$}}\\
$\dive\vE=0$, $\curl\vE=\mathbf 0$ & $\curl\vB=\mathbf 0$, $\dive\vB=0$ \\
$\vE=-\nabla\Phi_e$, $\nabla^{2}\Phi_e=0$ & $\vB=-\nabla\Psi_B$, $\nabla^{2}\Psi_B=0$ locally \\
\addlinespace[2pt]
\multicolumn{2}{@{}l}{\emph{Compensated Static Compact Bodies}}\\
$\vE=\mathbf 0$, $\vD=\vP$ & $\vB=\mathbf 0$, $\vH=-\vM$ \\
\botrule
\end{tabularx}
\end{table}
\FloatBarrier
\newpage

\section{Conclusion}
Compensated free and bound sources separate the auxiliary source equations from the total sources that determine $\vE$ and $\vB$. Exact static cancellation can produce $\vE=\mathbf 0$ with $\vD=\vP$, or $\vB=\mathbf 0$ with $\vH=-\vM$. The surviving auxiliary field is fixed by the prescribed polarization or magnetization and by the boundary data. It carries no independent degrees of freedom in vacuum.
\\

Time dependence requires a distinction between charge cancellation and complete four-current cancellation. The condition $\rhotot=0$ removes the Coulomb field and gives $\dive\Jtot=0$, but it does not force $\Jtot$ to vanish. A residual divergence-free current can radiate when its on-shell transverse transform is nonzero. For the tangential-sheet sphere, that transform is proportional to $j_2(kR)$. The resulting field is an ordinary transverse electromagnetic wave. Its amplitude is exact for arbitrary $kR$ at leading order in $1/r$, and its power is given by Eq.~\eqref{eq:power}. The full exterior field equals a point-dipole field with the effective moment $2\pi P_0R^{3}j_2(kR)$, near zone included.
\\

The same source has vanishing charge density, vanishing charge multipoles, and zero ordinary magnetic dipole moment at every instant. It nevertheless radiates at generic frequencies because finite-size structure remains in the total current. At the nonzero roots of $j_2(kR)$, the on-shell coefficient vanishes in every direction and the full outgoing exterior field is zero. By contrast, complete four-current cancellation sets $\rhotot=0$ and $\Jtot=\mathbf 0$ distributionally, so all retarded source-generated fields vanish. Within the compactly supported separable class of Proposition~\ref{prop:unique}, that complete cancellation is the unique choice that is nonradiating at every frequency. Since $\vD=\varepsilon_0\vE$ and $\vH=\vB/\mu_0$ in vacuum, neither case permits auxiliary-only radiation.

\FloatBarrier
\backmatter
\begin{appendices}

\section{Derivation of the Current Transform, Eq.~(\ref{eq:Jq})}
\label{app:jq}
The required spherical Bessel functions are
\begin{equation}
j_0(x)=\frac{\sin x}{x},\qquad
j_1(x)=\frac{\sin x}{x^{2}}-\frac{\cos x}{x},\qquad
j_2(x)=\Bigl(\frac{3}{x^{3}}-\frac{1}{x}\Bigr)\sin x-\frac{3\cos x}{x^{2}},
\end{equation}
with $j_0'=-j_1$, $j_2(x)=3j_1(x)/x-j_0(x)$, and $(t^{2}j_1(t))'=t^{2}j_0(t)$. At the origin, $j_0(0)=1$, $\lim_{x\to0}j_1(x)/x=1/3$, and $j_2(x)=x^{2}/15+O(x^{4})$, so every expression below has a regular limit as $x\to0$. Let $\mathbf a=\zhat$, $x=qR$, and $\mathbf n=\vr/R$ on the sphere $r=R$. Every transform below is taken per unit $g(t)$, and both source terms use $g(t)=\dd P_s/\dd t$. Equation~\eqref{eq:transform} defines the cosine-transform convention for the spatial profile $\vj$. The sine transform vanishes by parity. The volume indicator $\chi$ is even under $\vr\to-\vr$. The sheet profile is also unchanged because $(\theta,\phi)\to(\pi-\theta,\phi+\pi)$ leaves $\sin\theta\,\thetahat$ invariant.

Two angular integrals are needed. The scalar integral is
\begin{equation}
I_0(x)\equiv\oint_S \cos(x\,\qhat\cdot\mathbf n)\dd S
=2\pi R^{2}\!\int_{-1}^{1}\cos(xu)\dd u=4\pi R^{2}\,j_0(x).
\end{equation}
Symmetry fixes the tensor integral in the form $I_{ij}\equiv\oint_S n_in_j\,\cos(x\,\qhat\cdot\mathbf n)\dd S=\alpha\,\delta_{ij}+\beta\,\qhat_i\qhat_j$ because $\qhat$ is the only preferred direction. Its trace gives $3\alpha+\beta=I_0=4\pi R^{2}j_0$, while contraction twice with $\qhat$ gives
\begin{equation}
\alpha+\beta=2\pi R^{2}\!\int_{-1}^{1}u^{2}\cos(xu)\dd u
=-4\pi R^{2}\,j_0''(x)=4\pi R^{2}\Bigl[j_0(x)-\frac{2j_1(x)}{x}\Bigr].
\end{equation}
Using $j_0''=2j_1/x-j_0$, solving for $\alpha$ and $\beta$, and applying $j_0-3j_1/x=-j_2$ gives
\begin{equation}
I_{ij}(x)=4\pi R^{2}\Bigl[\frac{j_1(x)}{x}\,\delta_{ij}-j_2(x)\,\qhat_i\qhat_j\Bigr].
\end{equation}

The volume profile is uniform and equal to $\mathbf a$ inside the ball. Since the angular average of $\cos(\vq\cdot\vr)$ is $j_0(qr)$,
\begin{equation}
\JV_{\mathrm{vol}}(\vq)=\mathbf a\int_{r<R}\cos(\vq\cdot\vr)\dd V
=\mathbf a\,\frac{4\pi}{q^{3}}\int_0^{x}t^{2}j_0(t)\dd t
=4\pi R^{3}\,\frac{j_1(x)}{x}\,\mathbf a,
\label{eq:ballt}
\end{equation}
where $(t^{2}j_1)'=t^{2}j_0$ was used. For the sheet, $\zhat=\cos\theta\,\rhat-\sin\theta\,\thetahat$ gives $\sin\theta\,\thetahat=(\mathbf a\cdot\mathbf n)\mathbf n-\mathbf a$. The sheet profile in Eq.~\eqref{eq:Kf} can therefore be written as $\vk=\tfrac12R[(\mathbf a\cdot\mathbf n)\mathbf n-\mathbf a]$, and its transform is
\begin{align}
\JV_{\mathrm{sheet}}(\vq)
&=\tfrac{1}{2}R\bigl[\,I_{ij}a_j-I_0\,a_i\,\bigr]
=2\pi R^{3}\Bigl[\Bigl(\frac{j_1}{x}-j_0\Bigr)\mathbf a-j_2\,(\mathbf a\cdot\qhat)\,\qhat\Bigr]\nonumber\\
&=2\pi R^{3}\Bigl[\Bigl(j_2-\frac{2j_1}{x}\Bigr)\mathbf a-j_2\,(\mathbf a\cdot\qhat)\,\qhat\Bigr],
\end{align}
where the recurrence relation was used again. Adding Eq.~\eqref{eq:ballt} cancels the terms proportional to $j_1/x$ and leaves
\begin{equation}
\JV(\vq)=2\pi R^{3}\,j_2(x)\,\bigl[\mathbf a-(\mathbf a\cdot\qhat)\,\qhat\bigr],
\end{equation}
which is Eq.~\eqref{eq:Jq} for $\vq\neq\mathbf 0$. The value at $\vq=\mathbf 0$ follows from the limit. Two checks are useful. Along the axis, $\mathbf a-(\mathbf a\cdot\qhat)\qhat=\mathbf 0$. The sheet transform there is $-4\pi R^{3}[j_1(x)/x]\mathbf a$, using $\int_{-1}^{1}(1-u^{2})\cos(xu)\dd u=4j_1(x)/x$, and it exactly cancels the volume term. As $\vq\to\mathbf 0$, $j_2\to x^{2}/15$. The quantity $\JV(\mathbf 0)=\int\vj\dd V+\oint\vk\dd S$, multiplied by $g(t)=\dd P_s/\dd t$, equals the time derivative of the total dipole moment. Indeed, $\dd\vp/\dd t=\int\vr\,\partial_t\rho\,\dd^{3}r=-\int\vr\,(\dive\vJ)\,\dd^{3}r=\int\vJ\,\dd^{3}r$, since the boundary term vanishes at spatial infinity for a localized source. Here $\dd(\vp_b+\vp_f)/\dd t=\mathbf 0$, consistent with Eq.~\eqref{eq:Jq} vanishing at $q=0$.

\end{appendices}

\FloatBarrier

\section*{Statements and Declarations}
\noindent\textbf{Funding}\quad No funding was received for this study.

\medskip
\noindent\textbf{Competing Interests}\quad The author declares no competing interests.

\medskip
\noindent\textbf{Ethics Approval and Consent to Participate}\quad Not applicable.

\medskip
\noindent\textbf{Consent for Publication}\quad Not applicable.

\medskip
\noindent\textbf{Data Availability}\quad No datasets were generated or analyzed. The figures are schematic or are direct numerical plots of the closed-form expressions given in the manuscript.

\medskip
\noindent\textbf{Materials Availability}\quad Not applicable.

\medskip
\noindent\textbf{Code Availability}\quad A Python script was used only to prepare the figures. The definitions, derivations, inputs, and conclusions are fully stated in the manuscript and do not depend on access to the script.

\medskip
\noindent\textbf{Author Contributions}\quad N.R. conceived the study, completed the analytical derivations, prepared the figures, and wrote the manuscript.
\newpage

\FloatBarrier

\end{document}